\documentclass[a4paper,pdftex,11pt]{article}

\usepackage{amsthm}		
\usepackage{color,xspace}
\usepackage{stackrel}

\newtheorem{lemma}{Lemma}
\newtheorem{proposition}{Proposition}

\newtheorem{example}{Example}

\newcommand{\cL}{{\cal L}}
\newcommand{\cP}{{\cal P}}
\newcommand{\cW}{{\cal W}}
\newcommand{\WMG}{\textsf{WMG}\xspace}
\newcommand{\SSP}{\textsc{SSP}\xspace}
\newcommand{\remove}[1]{}

\begin{document}

\title{Dimension and codimension of simple games}
\author{Sascha Kurz\\
Department of Mathematics\\
University of Bayreuth\\
Germany\\
sascha.kurz@uni-bayreuth.de
\and
Xavier Molinero\thanks{Partially funded by grant MTM2012-34426/FEDER from MINECO and FEDER.}\\
Department of Mathematics\\
Universitat Polit\`ecnica de Catalunya\\
Spain\\
xavier.molinero@upc.edu
\and
Martin Olsen\\
Department of Business Development and Technology\\
BTECH, Aarhus University\\
Denmark\\
martino@btech.au.dk
\and
Maria Serna\thanks{Partially funded  by grant TIN2013--46181-C2-1-R from MINECO and FEDER
and by grant 2014SGR1034 (ALBCOM) of ``Generalitat de Catalunya''.}\\
Department of Computer Science\\
Universitat Polit\`ecnica de Catalunya\\
Spain\\
mjserna@cs.upc.edu
}

\date{}

\maketitle

\begin{abstract}
\noindent
This paper studies the complexity of computing a representation of a simple game
as the intersection (union) of weighted majority games,
as well as, the dimension or the codimension.
We also present some examples with linear  dimension and exponential codimension with respect to  the number of players.

\medskip

\noindent
\textbf{Keywords:}  Simple games, Dimension, Codimension, Computational complexity\\
\textbf{MSC:} 91B12, 91A12 \\
\end{abstract}

\section{Introduction and preliminaries}
We consider the so-called \emph{simple games} and
the computational complexity of representing them as unions or intersections of weighted majority games.
Simple games and its dimension, as well as, weighted majority games, were defined by Taylor and Zwicker~\cite{TaZw99}.
Later, Freixas and Marciniak~\cite{FM10} introduced a new concept, the codimension of simple games.

A {\em simple game} is a tuple $\Gamma=(N,\cW)$, where $N$ is a finite set of {\em players} and $\cW\subseteq\cP(N)$ is a monotonic family of subsets of $N$.
Furthermore, \emph{its dual} $\Gamma^*=(N,\cW^*)$ is the game such that
$\cW^*=\{S\subseteq{}N\,:\,N\setminus{}S\not\in{}\cW\}$.
$\Gamma$ is said to be \emph{self-dual} if $\Gamma=\Gamma^*$.
Note that $(\Gamma^*)^*=\Gamma$.
Given two simple games $\Gamma_1=(N_1,\cW_1)$ and $\Gamma_2=(N_2,\cW_2)$, they are \emph{equivalent}
if $N_1=N_2$ and $\cW_1=\cW_2$.
The subsets of $N$ are called {\em coalitions}, the set $N$ is the {\em grand coalition} and each $X\in\cW$ is a {\em winning coalition}.
The complement of the family of winning coalitions is the family of {\em losing coalitions} $\cL$, i.e., $\cL=\cP(N)\setminus\cW$.
Any of those set families determine uniquely the game $\Gamma$ and constitute one of the usual forms of representation for simple games~\cite{TaZw99}, although the size of the representation is not, in general, polynomial in the number of players~\cite{MRS15b}.

A simple game $\Gamma$ is a {\em weighted majority game} (\WMG) if it admits a representation by means of  $n+1$ nonnegative real numbers $[q;w_1,\ldots,w_n]$ such that $S\in\cW
\iff w(S)\ge{}q$ where, for each coalition $S\subseteq{}N$, $w(S)=\sum_{i\in{}S}w_i$. The number $q$ is called the quota and $w_i$ the weight of the player $i$. It is well known that any \WMG admits a representation with integer numbers.

The {\em dimension} of a simple game $\Gamma$ is the least $k$ such that there exists \WMG{}s $\Gamma_1,\ldots,\Gamma_k$ such that $\Gamma=\Gamma_1\cap\ldots\cap\Gamma_k$.
On the other hand, the {\em codimension} of a simple game $\Gamma$ is the least $k$ such that there exists \WMG{}s $\Gamma_1,\ldots,\Gamma_k$ such that $\Gamma=\Gamma_1\cup\ldots\cup\Gamma_k$.

Many theoretical results and examples about dimension and codimension~\cite{OlKuMo16,KuNa15,FM10,FrPu08,FP01,TaZw99,FrPu98}
have been constantly appearing during the last years,
as well as, computational complexity results~\cite{MOS16,FMOS11,DeWo06}.
We present some results that will be used in Section~\ref{sec:comcom}.
\begin{lemma} 
  \label{lemma_upper_bound_maximal_losing}
  The dimension of a simple game $v$ is bounded above by $\left|\mathcal{L}^M\right|$ and the co-dimension is bounded above by
  $\left|\cW^m\right|$.
\end{lemma}
\remove{
\begin{proof}
  For each coalition $S\in \mathcal{L}^M$ we set $q^S=1$, $w_i^S=0$ for all $i\in S$ and $w_i^S=1$ otherwise. Note that
  $\mathcal{L}^M\neq \emptyset$ since $\emptyset$ is a losing coalition. With this $w^S(S)=0<q^S$. However, for all $T\subseteq N$
  with $T\not\subseteq S$ we have $w(T)\ge 1=q^S$. Thus, we have $v=\bigwedge_{S\in\mathcal{L}^M} \left[q^S;w_1^S,\dots,w_n^S\right]$.
  Similarly, for each $S\in\cW^n$ we set $\widetilde{q}^S=|S|$, $\widetilde{w}_i^S=1$ for all $i\in S$ and $w_i^S=0$ otherwise.
  Note that $\cW^m\neq \emptyset$ since $N$ is a winning coalition. With this $\widetilde{w}^S(S)=\widetilde{q}^S$. However, for all
  $T\subseteq N$ with $S\not\subseteq T$ we have $w(T)<\widetilde{q}^S$. Thus, we have
  $v=\bigvee_{S\in\cW^m} \left[\widetilde{q}^S;\widetilde{w}_1^S,\dots,\widetilde{w}_n^S\right]$.
\end{proof}
}

\begin{lemma}
\label{lem:game-dual}
Let $\Gamma$ be a simple game.
 $\Gamma$ is the intersection of $t$ weighted games if and only if
 $\Gamma^*$ is the union of $t$ weighted games.
Furthermore a representation, as union (intersection), of $\Gamma^*$ can be obtained from a representation, as intersection (union), of $\Gamma$  in polynomial time.
 Moreover, dim($\Gamma$)=codim($\Gamma^*$),
 and if $\Gamma$ is self-dual then $dim(\Gamma)=codim(\Gamma)$.
\end{lemma}

Note that the converse statement of the last sentence is not true in general as there are weighted games which are not self-dual.


\section{Computational complexity of related problems}
\label{sec:comcom}

First, we present a simple game with $2n$ players, dimension $n$ and codimension $2^{n-1}$.
Other examples of simple games with $2n$ players and dimension $2^{n-1}$ can be found in~\cite{TaZw99,OlKuMo16}.

\begin{example}
\label{ex:dim-codim}
Given a positive integer $n$,
Freixas and Marcinicak (Theorem 2 of~\cite{FM10}) 
define a simple game
with $2n$ players and dimension $n$.
Let $\Gamma=(N,\cW)$ be a simple game
defined by $N=\{1,2,\ldots,2n\}$ and $S\in{}\cW$ iff
$S\cap\{2i-1,2i\}\neq\emptyset$, $i\in\{1,\ldots,n\}$,
then $\Gamma$ has dimension $n$, 
\remove{
\[
\begin{array}{rl}
 \Gamma = & \left[1; 1,1,0,0,\ldots,0,0\right] \cap \left[1; 0,0,1,1,0,0,\ldots,0,0\right] \cap \ldots \cap\\[.25cm]
          & \left[1;0,0,\ldots,0,0,1,1,0,0\right] \cap \left[1; 0,0,\ldots,0,0,1,1\right]
\end{array}
\]
}
\[
 \Gamma = \left[1; 1,1,0,0,\ldots,0,0\right] \cap \left[1; 0,0,1,1,0,0,\ldots,0,0\right]
          \cap \ldots \cap \left[1; 0,0,\ldots,0,0,1,1\right]
\]
As $S$ is a winning coalition in $\Gamma^*$ iff
$N\setminus{}S$ is a losing coalition in $\Gamma$, 
$ \Gamma^*=(N, \cW_1^* \cup \ldots \cup \cW_n^*)$,
where $\cW_i^*=\{S\subseteq{}N\,:\, \{2i-1,2i\}\subseteq{}S\}$, $i\in{}\{1,\ldots,n\}$.
As $\Gamma^*$ is a composition of $n$ unanimity games,
$\Gamma^*$ has dimension $2^{n-1}$~\cite{FP01} and
$\Gamma$ has codimension $2^{n-1}$ (by Lemma~\ref{lem:game-dual}).
\end{example}

\begin{proposition}
Given a simple game $\Gamma$ as union (intersection) of weighted games,
computing a representation of $\Gamma$ as intersection (union) of weighted games requires exponential time.
\end{proposition}

The complexity of 
several problems about representations of simple games as intersections of \WMG{}s were analyzed in~\cite{DeWo06}. We provide here a new reduction from the  NP-hard \emph{Subset Sum Problem (\SSP)}. Our reduction differs in the fact that  for the associated game  we know both the dimension and the codimension.  Recall that in the \SSP on input  
 a finite set $A=\{a_1,\ldots,a_n\}$ of positive integers and a positive integer $b$ we want to know whether  there is a subset $A'\subseteq{}A$ such that $\sum_{a_i\in{}A'}a_i=b$.

Let  $I=(b; a_1, a_2, \ldots, a_n)$ be an instance of the \SSP, $d>1$,
and let $\Gamma(I,d)$ be the  game defined on $n+d$ players by the intersection of the $d$ \WMG{}s:
\[
 \begin{array}{l}
  [3b+1; 3 a_1, \ldots, 3 a_n, 1, 1, 0, \stackbin{2(d-1)}{\ldots}, 0],\
  [3b+1; 3 a_1, \ldots, 3 a_n, 0, 0, 1, 1, 0, \stackbin{2(d-2)}{\ldots}, 0],\\
  \ldots,\  [3b+1; 3 a_1, \ldots, 3 a_n, 0, \stackbin{2(d-1)}{\ldots}, 0, 1, 1].
 \end{array}
\]




\begin{lemma} 
\label{prop:inter}\label{lem:red}
Let $d>1$. When  $I$ is a yes instance of \SSP then $dim(\Gamma(I,d))=d$ and $codim(\Gamma(I,d))=2^{d}$, otherwise,   $dim(\Gamma(I,d))=codim(\Gamma(I,d))=1$.
\end{lemma} 
\begin{proof}
Let $X\subset\{1,\ldots,n\}$,  $E=\{n+1,\dots, n+2d\}$ and $Y\subset E$.  Let $\Gamma$ be the game given in Example 1 taking $E$ as the set of players.
Observe that if $\sum_{i\in{}X}a_i> b$,  $\sum_{i\in{}X}a_i\ge{}b+1$, thus  $X \cup Y$ is a winning coalition in  $\Gamma(I,d)$. 
When $\sum_{i\in{}X}a_i< b$, $\sum_{i\in{}X}a_i \le{}b-1$, thus  $X \cup Y$ is a losing coalition in  $\Gamma(I,d)$. 
In the case that  $\sum_{i\in{}X}a_i= b$,   $X \cup Y$ is a winning coalition in  $\Gamma(I,d)$ iff $Y$ is a winning coalition in $\Gamma$.

When  $I$ is a yes instance of \SSP it follows that  $dim(\Gamma(I,d)) =d$ as otherwise $\Gamma$ will have a smaller dimension. On the other hand, it is easy to see that $\Gamma(I,d)^*$ is the composition of $d+1$ unanimity games and therefore $dim(\Gamma(I,d)^*) =2^d$ according to  \cite{FP01}, thus $codim(\Gamma(I,d)) =2^d$. 
When $I$ is a  no instance, there is  no $X\subseteq \{1,\dots,n\}$  for which  $\sum_{i\in{}X}a_i= b$. Therefore  $X\cup Y$ is winning in $\Gamma(I,d)$ iff $\sum_{i\in{}X}a_i\geq  b+1$. 
Thus $\Gamma(I,d)$ is the game $[b+1; a_1,\dots,a_n, 0, \stackbin{2d}{\ldots}, 0]$ and we conclude that $dim(\Gamma(I,d))=codim(\Gamma(I,d))=1$.
\end{proof}

Combining lemmas~\ref{lem:red} and \ref{lem:game-dual}   we can prove  the following results.

\begin{proposition}
\label{prop:inter}
Let $d_1$ and $d_2$ be two integers with $1\le{}d_2<d_1$.
Then the problem of deciding whether the union of $d_1$ given \WMG{}s can also be represented as the union of $d_2$ \WMG{}s is NP-hard.
\end{proposition}

\begin{proposition}
\label{prop:codim1}
Let $d_1$ and $d_2$ be two integers with $1\le{}d_1,d_2$.
Then the problem of deciding whether the intersection (union) of $d_1$ given \WMG{}s can also be represented
as the union (intersection) of $d_2$ \WMG{}s is NP-hard.
\end{proposition}

As a consequence of the previous results, given a simple game $\Gamma$ as union or intersection of \WMG{}s,
to compute $dim(\Gamma)$, $codim(\Gamma)$ or deciding whether $\Gamma$ is weighted are  NP-hard problems.

Recall that two game representations are said to be \emph{equivalent} whenever  the represented games have the same set of winning coalitions.
We can extend several results on equivalence problems from ~\cite{ElkindGGW08} to games given as unions of \WMG, in particular we have.

\begin{proposition}
\label{theo:Elkind4}
Checking whether a given  union of \WMG{}s
is equivalent to a given  union of \WMG{}s is co-NP-complete,
even if all weights are equal to 0 or 1. 
\end{proposition}

\section{Future work}
It remains open to exhaustively classify the dimension and codimension of all complete simple game  up to $n$ players.
Some bounds about dimension are given by Freixas and Puente~\cite{FP01} and Olsen \emph{et al.}~\cite{OlKuMo16}.
As well as to find complete simple games with \emph{small} dimension (codimension), but
with \emph{large} codimension (dimension), and to construct analytical examples with specific dimension and codimension.
It is also interesting to find real simple games with \emph{large} dimension or codimension as the example given by Kurz and Napel~\cite{KuNa15}.

%

\section*{References}
\begin{thebibliography}{10}

\bibitem{DeWo06}
V.G. De\u{\i}neko and G.J. Woeginger.
\newblock On the dimension of simple monotonic games.
\newblock {\em European Journal of Operational Research}, 170:315--318, 2006.

\bibitem{ElkindGGW08}
E.~Elkind, L.A. Goldberg, P.W. Goldberg, and M.~Wooldridge.
\newblock On the dimensionality of voting games.
\newblock In {\em AAAI Conference on
  Artificial Intelligence}, pages 69--74, 2008.



\bibitem{FM10}
J.~Freixas and D.~Marciniak.
\newblock On the notion of dimension and codimension of simple games.
\newblock {\em Cont. to Game Theory and Management}, 3:67--81, 2010.


\bibitem{FMOS11}
J.~Freixas, X.~Molinero, M.~Olsen, and M.~Serna.
\newblock On the complexity of problems on simple games.
\newblock {\em RAIRO - Op. Research}, 45:295--314, 9 2011.

\bibitem{FP01}
J.~Freixas and M.~A. Puente.
\newblock A note about games-composition dimension.
\newblock {\em Discrete Applied Mathematics}, 113(2--3):265--273, 2001.

\bibitem{FrPu98}
J.~Freixas and M.A. Puente.
\newblock Complete games with minimum.
\newblock {\em Annals of Operations Research}, 84:97--109, 1998.

\bibitem{FrPu08}
J.~Freixas and M.A. Puente.
\newblock Dimension of complete simple games with minimum.
\newblock {\em European Journal of Operational Research}, 188(2):555--568,
  2008.

\bibitem{KuNa15}
S.~Kurz and S.~Napel.
\newblock Dimension of the Lisbon voting rules in the EU Council: a challenge
  and new world record.
\newblock {\em Optimization Letters}, to appear, 2015.

\bibitem{MOS16}
X.~Molinero, M.~Olsen, and M.~Serna.
\newblock On the complexity of exchanging.
\newblock {\em Information Processing Letters}, to appear, 2016.

\bibitem{MRS15b}
X.~Molinero, F.~Riquelme, and M.~J. Serna.
\newblock Forms of representations for simple games: sizes, conversions and
  equivalences.
\newblock {\em Mathematical Social Sciences}, 76:87--102, 2015.

\bibitem{OlKuMo16}
M.~Olsen, S.~Kurz, and X.~Molinero.
\newblock On the construction of high dimensional simple games.
\newblock {\em CoRR}, abs/1602.01581, 2016.

\bibitem{TaZw99}
A.D. Taylor and W.S. Zwicker.
\newblock {\em Simple games: desirability relations, trading, and
  pseudoweightings}.
\newblock Princeton University Press, New Jersey, USA, 1999.

\end{thebibliography}

\end{document}